\documentclass[twocolumn,pra]{revtex4-2}

\usepackage{amsmath}
\usepackage{amssymb}
\usepackage{amsthm}
\usepackage{physics}
\usepackage{tikz}
\usepackage{tikzit}
\usepackage{circuitikz}
\usepackage{hyperref}
\usepackage{braket}

\input{styles.tikzdefs}

\tikzstyle{discarding}=[fill=white, draw=black, shape=circle, style=upground]
\tikzstyle{wide_discarding}=[fill=white, draw=black, shape=circle, style=upground, yscale=1.3]
\tikzstyle{smalldiscarding}=[fill=white, draw=black, style=upground, scale=0.75]
\tikzstyle{backdiscard}=[fill=white, draw=black, shape=circle, style=downground, scale=0.5]
\tikzstyle{smallbackdiscard}=[fill=white, draw=black, shape=circle, style=downground, scale=0.75]
\tikzstyle{state}=[fill=white, draw=black, style=triang, tikzit shape=rectangle]
\tikzstyle{effect}=[fill=white, draw=black, style=triangdag]
\tikzstyle{black_dot}=[style=dot, fill=black]
\tikzstyle{white_dot}=[style=dot, fill=white]
\tikzstyle{qblack_dot}=[style=ddot, fill=black]
\tikzstyle{qwhite_dot}=[style=ddot, fill=white]
\tikzstyle{whitephase}=[style=wphase dot, fill=white]

\tikzstyle{dottededge}=[-, dash pattern=on 2pt off 1pt]
\tikzstyle{thin_edge}=[-, style=thin]
\tikzstyle{thick_edge}=[-, style=thick]
\tikzstyle{thicker_edge}=[-, style=thicker]
\tikzstyle{double edge}=[-, style=doubled, draw=black, tikzit draw={rgb,255: red,18; green,168; blue,191}]
\tikzstyle{arrow}=[->]
\tikzstyle{new edge style 1}=[-, draw={rgb,255: red,242; green,233; blue,206}, fill={rgb,255: red,242; green,233; blue,206}]
\tikzstyle{morphism_shade}=[-, draw=black, fill={rgb,255: red,242; green,233; blue,206}, line join=bevel]
\tikzstyle{supermap_shade}=[-, fill={rgb,255: red,216; green,215; blue,242}, draw=black, line join=bevel]
\tikzstyle{hole_shade}=[-, fill=white, draw=black, line join=bevel]
\tikzstyle{new edge style 2}=[-, draw={rgb,255: red,14; green,188; blue,83}]
\tikzstyle{green}=[-, fill=none, draw={rgb,255: red,0; green,106; blue,106}]
\tikzstyle{green_dotted}=[-, draw={rgb,255: red,0; green,106; blue,106}, dash pattern=on 1pt off 0.7pt]

\usetikzlibrary{circuits.ee.IEC}
\newcommand{\Ground}{%
  \mathbin{\text{\begin{tikzpicture}[circuit ee IEC,yscale=0.6,xscale=0.6,baseline=0ex]
  \draw (0,-0.5ex) to (0,2.5ex) node[ground,rotate=90,xshift=.75ex] {};
  \end{tikzpicture}}}%
}

\newcommand{\cat}[1]{\mathcal{#1}}
\newcommand{\comb}{\mathsf{Comb}}
\newcommand{\qbox}{\mathsf{QBox}}
\newcommand{\morph}[1]{\xrightarrow{#1}}
\newcommand{\fhilb}{\mathsf{FHilb}}
\newcommand{\CP}{\mathsf{CP}}
\newcommand{\CPTP}{\mathsf{CPTP}}
\newcommand{\mat}[1]{\mathsf{Mat}_{#1}}
\newcommand{\stoch}{\mathsf{Stoch}}
\newcommand{\Split}{\mathsf{Split}}
\newcommand{\hypdec}{\mathtt{hypdec}}

\newcommand{\reals}{\mathbb{R}}

\newcommand{\C}{\cat{C}}
\newcommand{\D}{\cat{D}}

\theoremstyle{definition}
\newtheorem{defn}{Definition}
\theoremstyle{plain}

\theoremstyle{plain}
\newtheorem{lem}{Lemma}
\theoremstyle{plain}
\newtheorem{thm}{Theorem}
\theoremstyle{plain}

\theoremstyle{remark}
\newtheorem*{remark}{Remark}
\theoremstyle{definition}
\newtheorem{example}{Example}

\begin{document}

\author{James Hefford}
\email{james.hefford@inria.fr}
\author{Matt Wilson}
\email{matthew.wilson@centralesupelec.fr}
\affiliation{Université Paris-Saclay, CNRS, ENS Paris-Saclay, Inria, CentraleSupélec, Laboratoire Méthodes Formelles}

\date{March 23, 2026}

\title{Decoherence to quantum theory from a causally-indefinite post-quantum theory}

\begin{abstract}
    We find a process satisfying the axioms of hyper-decoherence which produces standard quantum theory from the theory of quantum boxes (higher-order quantum theory with the non-signalling tensor product). 
    This hyper-decoherence map evades the no-go theorem of Lee and Selby \cite{lee_nogo} by relaxing constraints on signalling to the past and the uniqueness of purifications.
    We discuss some natural opposing conclusions: that the existence of this map might be evidence of a genuine hyper-decoherence process producing causal quantum theory from its causally-indefinite higher-order theory; or that it serves as an indication that the axioms of hyper-decoherence might need careful re-consideration, especially regarding the subtle albeit central role that purity plays.
\end{abstract}

\maketitle

\section{Introduction}
One way to understand the emergence of classical theory from quantum theory is via the process of decoherence, in which the limited power or control of observers forbids them direct access to the quantum realm. In the knowledge that quantum theory is not a complete description of our physical world (for instance in its lack of a satisfactory accommodation of gravity), one is naturally led to the question of how quantum theory could emerge from some yet-to-be-defined post-quantum theory, via an analogous process of \textit{hyper}-decoherence due to constraints on the powers of purely-quantum observers.

This idea appears to have been first explicitly discussed in \cite{zyczkowski} in relation to deriving quantum theory from minimal and physically reasonable axioms \cite{hardy}, and lies within the broader problem of singling quantum theory out amongst the larger class of Generalised Probabilistic Theories (GPTs) \cite{barrett_gpts}, Operational Probabilistic Theories (OPTs) \cite{chiribella_purification,chiribella_informational} or Categorical Probabilistic Theories (CPTs) \cite{gogioso_cpt}. 

The search for such post-quantum theories appears to be cut short by the no-go theorem of \cite{lee_nogo}, which establishes that there can be no such post-quantum theory satisfying both causality and the existence of unique purifications.
Naturally, one can question the validity of these assumptions in the hope of bypassing this theorem, and a number of partial toy theories have been suggested which exhibit some notion of hyper-decoherence \cite{zyczkowski,dakic_cubes,gogioso_hypercubes,hefford_hypercubes}, each one breaking, more or less severely, at least one axiom we would expect of a physically reasonable theory \cite{lee_interference} or of a physically reasonable hyper-decoherence map, thereby allowing them to side-step the no-go theorem.

These toy theories however, do not just break causality or uniqueness of purifications. Quartic Quantum Theory \cite{zyczkowski} and density cubes \cite{dakic_cubes} suffer from substantial foundational issues including ill-defined processes and joint systems \cite{lee_interference}.
Density hypercubes \cite{gogioso_hypercubes,hefford_hypercubes} on the other hand can be considered a legitimate theory in the sense that it is a CPT \cite{gogioso_cpt} with tomography, however, its hyper-decoherence process suffers from being not only non-causal but also non-deterministic \cite{hefford_hypercubes}.

In the elusive search for a satisfactory post-quantum theory, a more principled approach could be to examine toy models of aspects of quantum gravity, which we expect to be a theory more fundamental than traditional quantum theory. Notably, an influential proposal of \cite{Hardy_2007}, that theories which unify quantum theory with gravity ought also to exhibit non-fixed causal structures, appears to draw a striking parallel with the non-causal conclusion of the no-go of \cite{lee_nogo}. This in turn begs the question, \textit{could quantum theory with indefinite causal order be that elusive post-quantum theory which hyper-decoheres to traditional quantum theory?}.

In this article we argue the affirmative, by finding a map satisfying the axioms of hyper-decoherence from the \textit{theory of quantum boxes} $\qbox$ into standard quantum theory, where $\qbox$ is the most natural fragment of higher-order quantum theory for modelling indefinite causal order.
On the one hand, the transition from higher to lower order has a natural interpretation in terms of the reduction in power of observers.
\begin{figure}[h]
    \input{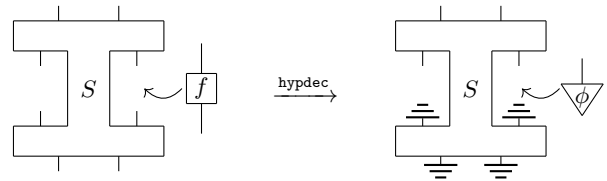} 
    \caption{Hyper-decoherence of causal quantum theory from causally-indefinite higher-order quantum theory.}
\end{figure}
On the other hand, if one finds this map to be too far outside of what might be expected for a map that resembles decoherence, we consider the alternative conclusion, that the axioms of hyper-decoherence might be incomplete, or at least, in need of refinement.

Whilst $\qbox$ is indeed non-causal, we highlight the fact that the essential feature of post-quantum theories for the ruling out of hyper-decoherence, is really the existence of \textit{unique} purifications.
The theory $\qbox$ does indeed support purifications (in the sense of convex geometry) for general states, however, those purifications are not unique up to a reversible transformation on the environment.
This occurs because any channel which discards and prepares a pure state is extremal in the convex space of quantum channels.
Therefore, it is the tension in our understanding of purity within post-quantum theories which in turn allows for a way to evade to no-go theorem of \cite{lee_nogo}. 

\section{Process Theories}\label{sec:hopt}
We will develop the results of this paper in the setting of \textit{process theories} also known as symmetric monoidal categories (SMCs) \cite{coecke_kissinger}.
A process theory $\C$ consists of a collection of systems or \textit{objects} which we denote in capital $H,K$ etc.\ and a collection of processes or \textit{morphisms} for instance $f:H\morph{}K$.
These processes can be composed in sequence $gf$ and in parallel by tensor product $g\otimes f$ with these two compositions compatible in the expected way.
Each system $H$ comes with an identity process $1_H$ with the property $1_K f = f= f 1_H$ and it is possible to invertibly swap the order of systems in a tensor product $\sigma_{H,K}:H\otimes K\morph{}K\otimes H$.
There is a trivial system $I$ with the property that $H\otimes I \cong H$ for every system $H$.

\begin{example}[Pure Quantum Theory]
    There is a process theory $\fhilb$ whose systems are the finite dimensional Hilbert spaces $H$ and whose processes are the linear maps.
    The parallel composition is the usual tensor product.
\end{example}

\begin{example}[Mixed Quantum Theory]
    There is a process theory $\CP$ whose systems are the finite dimensional $C^*$-algebras and whose processes are the completely positive maps.
    The tensor product is the standard one.
    There is a sub-process theory $\CPTP$ of only the completely positive trace preserving maps.
\end{example}

\begin{example}[Classical Theory]
    There is a process theory $\mat{\reals^+}$ whose systems are the natural numbers $n\in\mathbb{N}$ and whose processes $m\morph{}n$ are the $m\times n$ matrices with entries from $\reals^+$.
    Sequential composition of matrices is given by matrix multiplication and parallel composition by Kronecker product.
    There is a sub-process theory $\stoch$ of only the stochastic matrices.
\end{example}

\begin{defn}
    Given a process theory $\C$ there is another process theory $D(\C)$ whose systems are formal tensor products of pairs $\otimes_{i=1}^n [H_i,H_i]$ where each $H_i$ is a system of $\C$.
    A process $f:\otimes_i[H_i,H_i]\morph{}\otimes_j[K_j,K_j]$ is a process $f:\otimes_i(H_i\otimes H_i)\morph{} \otimes_j (K_j\otimes K_j)$ of $\C$.
    Composition, identities and the tensor are inherited from $\C$, with $[I,I]$ as the trivial system.
\end{defn}

We interpret a process $f:\otimes_i[H_i,H_i]\morph{}\otimes_j[K_j,K_j]$ in $D(\C)$ as a higher-order map taking a family a processes $\{H_i\morph{}H_i\}_i$ to a family $\{K_j\morph{}K_j\}_j$.
In this way $D(\C)$ forms the backbone of a process theory of higher-order maps over $\C$, though one which does not account for which higher-order maps are causally valid.
In the next section we will show how to equip $D(\C)$ with a notion of causality in order to restrict its maps to only those that are causally valid.

\section{Multi-Environment Structures}\label{sec:multi_env}

Causality in process theories can be captured by equipping them with an environment structure \cite{coecke_classicality,coecke_terminality}.

\begin{defn}
    An \textit{environment structure} for a process theory $\cat{C}$ consists of a choice of a discard effect $\Ground_H$ for each system $H$ such that the choices are closed under the tensor product.
    \begin{equation*}
        \input{figs/env_struct.tikz}
    \end{equation*}
\end{defn}

An environment structure serves to give a process theory a notion of causality: we say that a process $f:H\morph{}K$ is normalised or \textit{causal} if 
\begin{equation*}
    \input{figs/causal.tikz}
\end{equation*}
so that discarding the output $K$ of $f$ is the same as discarding the input $H$.
One can always restrict the a process theory $\C$ equipped with an environment structure to just the causal maps giving the subtheory $\C^{\Ground}$ of maps that can be made to happen with certainty.

\begin{example}
    Mixed quantum theory $\CP$ has an environment structure given by the trace for each system.
    The causal sub-process theory $\CP^{\Ground}$ is $\CPTP$.
    Classical theory $\mat{\reals^+}$ has an environment structure given by the row vector of all 1s, $\begin{pmatrix} 1 \dots 1 \end{pmatrix}:n\morph{}1$ for each $n$.
    The causal sub-process theory $\mat{\reals^+}^{\Ground}$ is $\stoch$.
\end{example}

Causality of a process theory imposes a more traditional notion of no-signalling in terms of the outcomes of Bell scenarios.
Given any bipartite state, with each half located in spacelike separated regions, one can ask whether deterministic events occurring in one region can influence the state of the other region.
In a causal process theory the answer is no, since for any state $\ket{\psi}$ and pair of processes $f,g$ we can see that 
\begin{equation*}
    \input{figs/no_super_1.tikz} \ =  \   \input{figs/no_super_2.tikz} \ =  \  \input{figs/no_super_3.tikz} 
\end{equation*}

For process theories of higher-order maps it is expected that systems will no longer possess a unique discarding effect.
For instance, in the theory of quantum combs a system $[H,H]$ can be discarded by any map which prepares a causal state on the bottom and traces out the top,
\begin{equation*}
    \input{figs/comb_environment.tikz} \qquad \text{where} \qquad \input{figs/causal_state.tikz} = 1
\end{equation*}
Thus one may worry that it is not possible to express no-signalling in such theories.
In fact, the uniqueness of the discard effect is not of paramount importance to the previous argument and one can directly formulate the principle of no superluminal signalling via bipartite states when there is more than one deterministic effect \cite{Wilson_2021}.
To do so, we will introduce the notion of a multi-environment structure to capture all the permitted ways of discarding a system.

\begin{defn}
    A \textit{multi-environment structure} $\Sigma$ for a process theory $\cat{C}$ consists of, for each system $H$, a family of effects $\Sigma_H = \{\Ground:H\morph{}I\}$ such that: 
    \begin{equation*}
        \input{figs/multi_env_struct1.tikz},
    \end{equation*}
    \begin{equation*}
        \Sigma_I = \left\{ \input{figs/multi_env_struct2.tikz} \ \right\},
    \end{equation*}
    \begin{equation*}
        \input{figs/multi_env_struct3.tikz}
    \end{equation*}
\end{defn}

\begin{remark}
    Our notion of multi-environment structure is related to the one of \cite{gogioso_cpm} but the symmetry conditions are distinct.
    Here we ask for invariance under the symmetries $\sigma$ of $\cat{C}$, whereas in \cite{gogioso_cpm} the symmetries are asked for \textit{internally} to each system.
\end{remark}
Note that a multi-environment structure with one effect for each system is just an ordinary environment structure.

We can then say that a theory with a multi-environment structure is non-signalling if and only if for every state $\ket{\psi}$ and every pair of discarding effects $\Ground_1, \Ground_2$ we have the following
\begin{equation*}
    \input{figs/no_super_5.tikz} \ =  \   \input{figs/no_super_4.tikz}
\end{equation*}
or equivalently, if for every state $\ket{\psi}$, discarding effect $\Ground$, and pair of processes $f,g$ 
\begin{equation*}
    \input{figs/no_super_3.tikz} \ =  \   \input{figs/no_super_1.tikz}
\end{equation*}

Note that a theory can be non-signalling, without imposing that every process is no-backwards in time signalling in the following sense.
\begin{defn}
    Let $\cat{C}$ be a process theory with a multi-environment structure.
    A process $f:H\morph{}K$ is \textit{no-backwards-signalling} if there exists $\Ground_H\in \Sigma_H$ such that for every $\Ground_K\in\Sigma_K$,
    \begin{equation*}
        \input{figs/backwards_sig.tikz}
    \end{equation*}
\end{defn}
We see that a process is no-backwards-signalling when for any permitted way of discarding its output, we get the same discard on its input and thus no deterministic event on the output can influence the state of the input.

With these formalities in place we can now generalise the notion of deterministic process to theories with a multi-environment structure.

\begin{defn}
    In a process theory $\cat{C}$ with a multi-environment structure we say that a process $f:H\morph{}K$ is \textit{deterministic} if for any system $L$, and any $\Ground_{K\otimes L}\in \Sigma_{K\otimes L}$,
    \begin{equation*}
        \input{figs/deterministic.tikz}.
    \end{equation*}
\end{defn}

In the case that there is one unique effect in $\Sigma_H$ for each $H$, the previous definition reduces to the usual notion of deterministic/causal/normalised process.

Deterministic maps are compositionally well-behaved.
It is fairly straightforward to see that the identity process on any system is deterministic and deterministic maps are closed under composition and tensor product.
Thus the deterministic maps from $\cat{C}$ form a sub-process theory which we denote $\cat{C}^{\Ground}$.

Let us now consider the theory $D(\C)$ which can be equipped with two different multi-environment structures leading to two different deterministic subtheories that are non-signalling, but in which not every process is no-backwards in time signalling.

\begin{example}\label{ex:comb_env}
    Consider the process theory $D(\C)$ where $\C$ is a process theory with an environment structure.
    $D(\C)$ has a multi-environment structure given by picking $\Sigma_{\otimes_i [H_i,H_i]}$ to be the set of higher-order maps given by discarding the top $H:=\otimes_i H_i$ and preparing a causal state on the bottom $H=\otimes_i H_i$.
    \begin{equation*}
        \Sigma_{\otimes_i [H_i,H_i]} := \left\{ \input{figs/comb_environment.tikz} \quad : \quad \input{figs/causal_state.tikz} = 1 \right\}
    \end{equation*}
    In the case $\C=\CP$, the process theory $D(\CP)^{\Ground}$ is that of deterministic quantum combs \cite{chiribella_networks}.
\end{example}
\begin{remark}
    For those familiar with the $\mathsf{Caus}$-construction \cite{kissinger_caus}, the previous example is equivalent to the sub-process theory of $\mathsf{Caus}(\CP)$ generated by the signalling tensor product $\rotatebox[origin=c]{180}{$\&$}$ on systems of type $[H,H]$ with each $H$ a first-order system.
    It is also equivalent to $\comb(\CPTP)$ as defined in \cite{hefford_coend,hefford_supermaps}.
\end{remark}

\begin{example}\label{ex:supermap_env}
    The process theory $D(\CP)$ has a multi-environment structure given by picking $\Sigma_{\otimes_i [H_i,H_i]}$ to be the set of process matrices $W$ \cite{oreshkov} on $\otimes_i [H_i,H_i]$.
    \begin{equation*}
        \Sigma_{\otimes_i [H_i,H_i]} := \left\{ \ \input{figs/multi_env_struct_proc_mat.tikz} \ \right\}
    \end{equation*}
    Note that this is \textit{not} the same multi-environment structure as in Example \ref{ex:comb_env}, in particular $\otimes_i[H_i,H_i]$ is generally not equal to $[\otimes_i H_i,\otimes_i H_i]$. 
    Moreover, the collection of discard maps on the latter coincides with those of Example \ref{ex:comb_env}, while there are more valid discard maps on the former given by the process matrices which are not of the form of a preparation-discard.
\end{example}

The multi-environment structure of the previous example leads to the definition of the theory $\qbox$ which will be the focus of the remainder of this article and from which we will show quantum theory can hyper-decohere.

\begin{defn}
    Equipping $D(\CP)$ with the multi-environment structure of Example \ref{ex:supermap_env} yields the process theory $\qbox := D(\CP)^{\Ground}$ of second-order quantum operations under the non-signalling tensor product as its deterministic subtheory.
    The systems are generated by taking arbitrary non-signalling tensor products of those of the form $[H,H]$ for a finite dimensional Hilbert space $H$, and thus take the form $\otimes_{i=1}^n [H_i,H_i]$.
    A process $S:\otimes_i [H_i,H_i]\morph{} \otimes_j [K_j,K_j]$ is a quantum supermap which takes non-signalling channels as its input and outputs a non-signalling channel.
    The tensor product of $\qbox$ is the non-signalling tensor and the trivial system is $[\mathbb{C},\mathbb{C}]$. 
\end{defn}

\begin{remark}
    For those familiar with the $\mathsf{Caus}$-construction \cite{kissinger_caus}, $\qbox$ is equivalent to the sub-process theory of $\mathsf{Caus}(\CP)$ generated by the non-signalling tensor product $\otimes$ on systems of type $[H,H]$ with each $H$ a first-order system.
\end{remark}

The process theory $\qbox$ contains all the quantum processes with indefinite causal order as defined in \cite{chiribella_supermaps,oreshkov}. 
Unsurprisingly, given the name used to refer to its tensor product, $\qbox$ is a non-signalling theory, a proof of this fact is given in Appendix \ref{app:superluminal}.

\section{Hyper-decoherence}\label{sec:hyperdec}

In this section we will consider how one process theory $\C$ might be contained in another $\D$ by a decoherence-like process.
There are a few properties we ought to expect of such a hyper-decoherence map based on the arguments of \cite{lee_nogo}, in particular,
\begin{itemize}
    \item[Ax1:]\label{ax:1} it is idempotent so that once hyper-decoherence has occurred any further hyper-decoherence leaves the systems invariant,
    \item[Ax2:]\label{ax:2} it is no-backwards-signalling, banning signalling from the future into the past,
    \item[Ax3:]\label{ax:3} it copreserves the purity of states in $\cat{C}$, so that any state which hyper-decoheres to a pure state of $\cat{C}$ must be pure in $\cat{D}$,
    \item[Ax4:]\label{ax:4} it preserves maximal mixtures, meaning that when hyper-decoherence is applied to a maximally mixed state of $\cat{D}$ it returns a maximally mixed state in $\cat{C}$.
\end{itemize}

The final two axioms perhaps require further explanation: their idea is to ban certain undesirable properties with regards to the purity and mixedness of states under hyper-decoherence.

\begin{defn}
    A deterministic state of $\D$ is \textit{pure} if it cannot be written as a convex combination of other distinct deterministic states of $\D$.
    Otherwise we say that the state is \textit{mixed}.
    Similarly, a deterministic state of the subtheory $\C$ is \textit{pure in $\C$} if it cannot be written as a convex combination of other deterministic states in $\C$. 
\end{defn}

If we view a pure state as a state of maximal knowledge, then if a mixed state, and thus a state of less-than-maximal knowledge, could become pure under hyper-decoherence we would have a rather odd situation in which hyper-decoherence leads to a gain in knowledge.
This is banned by axiom \hyperref[ax:3]{Ax3} requiring that pure states in $\C$ are also pure in $\cat{D}$.

\begin{defn}
    A deterministic state $\rho$ of $\D$ is \textit{maximally mixed} if every deterministic state of $\D$ appears in some convex decomposition of $\rho$ and if $\rho$ is invariant under invertible transformations of $\D$.
    A deterministic state of the subtheory $\C$ is \textit{maximally mixed in $\C$} if the same conditions hold replacing everywhere $\D$ with $\C$.
\end{defn}

Similarly, if the maximally mixed state in $\C$ was not the maximally mixed state in $\cat{D}$ we could map a state of minimal knowledge to one of greater knowledge under hyper-decoherence.
This is banned by \hyperref[ax:4]{Ax4}.

The final property of hyper-decoherence is that applying it to the systems and processes of $\D$ yields the systems and processes of $\C$.
This can be formalised using the \textit{idempotent splitting} or \textit{Karoubi envelope} of a process theory.
The idea is to produce a new process theory $\Split(\D)$ from $\D$ whose systems are the idempotents of $\D$.
This turns the decoherence maps in $\D$ into objects in $\Split(\D)$ which we can think of as the decohered systems with the maps between such systems compatible with the decoherence.
This method has been used extensively in the categorical quantum mechanics literature to model quantum-classical decoherence \cite{selinger_idempotents,heunen_completely,coecke_classicality,gogioso_thesis,gogioso_cpt} and extended to hyper-decoherence in \cite{gogioso_hypercubes,selby_thesis,hefford_thesis,hefford_galois}.

\begin{defn}
    Given a process theory $\cat{C}$, the \textit{idempotent splitting} $\Split(\cat{C})$ has systems of the form $(H,e)$ where $e:H\morph{}H$ is an idempotent.
    A process $f:(H,e)\morph{}(K,e')$ is a process $f: H\morph{}K$ of $\cat{C}$ that is invariant under the idempotents, $f = e'fe$.
    Composition and tensor product are inherited in the obvious way from $\cat{C}$, with the identity on $(H,e)$ given by $e$.
\end{defn}

So, given the process theory $\D$ we upgrade it to the process theory $\Split(\D)$ and then study the collection of systems given by the hyper-decoherence maps in $\Split(\D)$.
To show that this is equivalent to the desired process theory $\C$, we need the notion of an equivalence of process theories.

\begin{defn}
    Let $\C$ and $\D$ be process theories.
    An \textit{equivalence of process theories} consists of a map $F:\C\morph{}\D$ on systems and processes which preserves composition, tensor product and identity processes.
    Furthermore, $F$ must be a bijection on the sets of processes so that there is an isomorphism $\C(H,K) \cong \D(FH,FK)$ between processes $H\morph{}K$ in $\C$ and processes $FH\morph{}FK$ in $\D$.
    Finally, $F$ must be an essential surjection on the systems, that is for every system $K$ of $\D$, there exists a system $H$ of $\C$ and an isomorphism $FH\cong K$.
\end{defn}

\begin{remark}
    An an equivalence of process theories is more traditionally known as a monoidal equivalence of categories, that is a fully-faithful and essentially surjective on objects monoidal functor $F:\C\morph{}\D$.
\end{remark}

We are now in a position to define what it means for one process theory to hyper-decohere to another.
\begin{defn}
    A process theory $\cat{D}$ \textit{supports hyper-decoherence} to a process theory $\C$ if every system $H$ of $\D$ possesses a hyper-decoherence map $\hypdec:H\morph{}H$ satisfying axioms \hyperref[ax:1]{Ax1} - \hyperref[ax:4]{Ax4}, such that the full sub-process theory of $\Split(\D)$ spanned by systems of the form $(H,\hypdec)$ is equivalent to $\cat{C}$, and at least one of the hyper-decoherence maps is not the identity process.
\end{defn}

In the case that a process theory hyper-decoheres to quantum theory we will call the theory post-quantum.

\begin{defn}
    A process theory $\D$ is \textit{post-quantum} if it supports hyper-decoherence to the process theory $\CPTP$.
\end{defn}

\section{Hyper-decoherence of Quantum Theory from Higher-Order Quantum Theory}\label{sec:hyperdec_hoqt}

In this section we will prove our main result showing that quantum theory can hyper-decohere from the post-quantum theory $\qbox$ of second-order quantum operations.

The hyper-decoherence maps are given by the following processes $\otimes_i [H_i,H_i] \morph{} \otimes_i [H_i,H_i]$ in $\qbox$.
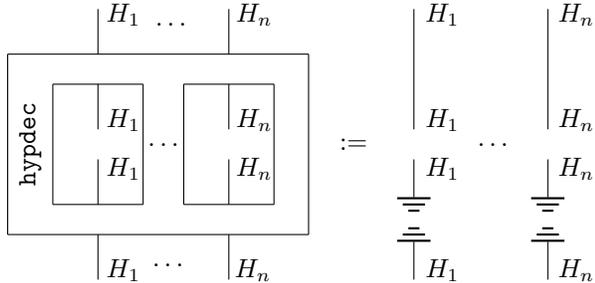
\begin{equation*}
    \input{figs/ico_hypdec.tikz} 
\end{equation*}
Explicitly, the process $\hypdec$ is given by the completely depolarising map on the bottom of the supermap and the identity process on the top.
This is a deterministic superchannel, meaning that $\hypdec(f_1,  \dots, f_n)$ is a CPTP map when applied to CPTP maps $f_1,  \dots, f_n$.
Consequently, it is a process in $\qbox$.

\begin{lem}
    $\hypdec$ is idempotent and no-backwards-signalling.
\end{lem}
\begin{proof}
Idempotency follows from checking the top and bottom of the hyper-decoherence map separately. The top is trivial since it is the identity process, the bottom follows easily by noting that
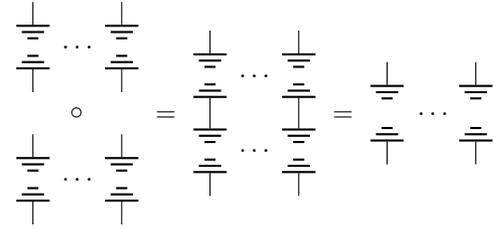
\begin{equation*}
    \input{figs/idem_1.tikz} \ = \  \input{figs/idem_2.tikz} \ = \  \input{figs/idem_3.tikz}
\end{equation*}
No-backwards-signalling follows since for any multi-environment element (i.e.\ process matrix) $W$,
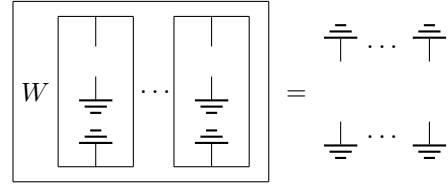
\begin{equation*}
    \input{figs/effect_1.tikz} \ = \  \input{figs/effect_3.tikz}
\end{equation*}
\end{proof}

\begin{lem}
    There is an equivalence of process theories between the full sub-process theory of $\Split(\qbox)$ spanned by systems of the form $(\otimes_i [H_i,H_i],\hypdec)$ and $\CPTP$.
\end{lem}
\begin{proof}
    A formal proof is in appendix \ref{app:equivalence}.
    The core idea is to map each higher-order map into a lower order one:
    \begin{equation*}
        \input{figs/equiv_proof1_ico.tikz}
    \end{equation*}
    Indeed, since this is the partial application of a superchannel on two CPTP maps, the result is automatically CPTP.
    $F$ can then be shown to satisfy all the requirements to make it an equivalence of process theories.
\end{proof}

\begin{lem}
    The pre-image of any pure state under hyper-decoherence is pure.
\end{lem}
\begin{proof}
    Applying hyper-decoherence to a state $S$ in $\qbox$ returns:
    \begin{equation*}
        \input{figs/purity_proof1_ico.tikz}
    \end{equation*}
    where $g$ arises from the isomorphism between states in $\qbox$ and multi-partite non-signalling channels.
    Under the equivalence with quantum theory, this becomes the quantum state
    \begin{equation*}
        \input{figs/purity_proof2.tikz}
    \end{equation*}
    where we have identified wires $1, \dots, m$ in the input and output of $g$.
    For this state to be pure, $g$ must take the form of a discard-and-prepare channel for a pure state. Indeed, to have that 
    \[ 
    g\left(\frac{I}{d}\right) = \ketbra{\phi}{\phi} 
    \]
    entails that for any orthonormal basis $\{\ket{e_i}\}_{i=1}^d$, we have
    \[
    \ketbra{\phi}{\phi} = g\left(\frac{I}{d}\right) = g\left(\sum_{i=1}^d \frac{\ketbra{e_i}{e_i}}{d} \right) = \sum_{i=1}^d \frac{g(\ketbra{e_i}{e_i})}{d}.
    \]
    The extremality of $\phi$, by definition, then gives that for all $i$,
    \[
    g(\ketbra{e_i}{e_i}) = \ketbra{\phi}{\phi}.
    \]
    Since this applies to any orthonormal basis, it applies to every state, and so $g$ is the constant preparation of the pure state $\ketbra{\phi}{\phi}$.

    Now, in order for $g$ to be pure post-quantumly, it must be an extremal channel.
    Note that $g$ has a decomposition,
    \begin{equation*}
        \input{figs/purity_proof3.tikz}
    \end{equation*}
    giving Kraus operators $K_i:=\ket{\phi}\bra{i}$.
    Note that the set $\{K_i^\dag K_j\}_{i,j}$ is linearly independent and so by Choi's extremality criterion \cite{choi}, $g$ is extremal and thus pure.
\end{proof}

\begin{lem}
    The maximally mixed state is preserved by hyper-decoherence.
\end{lem}
\begin{proof}
The maximally mixed state in $\qbox$ is given by \[      \input{figs/pq_maxmix.tikz},  \]
which is easily seen to be sent to itself by the hyper-decoherence map.
Under the equivalence of process theories between $\Split(\qbox)$ and $\CPTP$, this is the maximally mixed state \[      \input{figs/q_maxmix.tikz}  \] of quantum theory.
\end{proof}

Putting together the previous lemmas we can conclude the following theorem.

\begin{thm}
    $\qbox$ is a post-quantum theory.
\end{thm}

Returning to the no-go theorem of \cite{lee_nogo}, which states that there can be no post-quantum theory which is both causal and supports unique purifications, it is natural to ask how exactly $\qbox$ manages to support a hyper-decoherence.
The role of causality is rather minor in the proof of this no-go theorem, with its role being to provide the existence of some effect with respect to which purifications are expressed.
The no-go theorem does however directly lean on the existence of unique purifications, that is, the requirement that for every mixed state $\rho$ there exists a pure state $\psi$ such that 
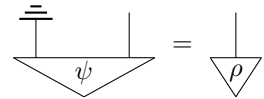
\begin{equation*}
    \input{figs/purif_1.tikz} \ =  \   \input{figs/purif_2.tikz}
\end{equation*}
with this pure state being unique up to a reversible transformation on the environment in the sense that for any pair of purifications $\psi_1 , \psi_2$ of $\rho$ there exists a reversible transformation $r$ such that 
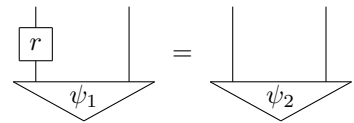
\begin{equation*}
    \input{figs/purif_3r.tikz} \ =  \   \input{figs/purif_4.tikz}
\end{equation*}

Since $\qbox$ is built as a theory in which states are quantum channels, it would be natural to imagine that the pure states of $\qbox$ are the pure quantum channels (i.e.\ the isometries or the unitaries), in which case purifications would be expected to be unique by the uniqueness of Stinespring dilations.
Interpreted as a generalised physical theory however, the appropriate notion of purity for states in $\qbox$ is convex extremality and dilations of CPTP maps to extremal CPTP maps are not unique.
We give a proof of this fact in Appendix \ref{app:purifications}.
It is then, the non-uniqueness of purifications for quantum theory with indefinite causal order which allows for it to bypass the no-go and support a hyper-decoherence into standard, causal, quantum theory.

\section{Discussion and Conclusion}

Whilst no causal theory with unique purifications can decohere to quantum theory \cite{lee_nogo}, we have shown that higher-order quantum operations via their non-uniqueness of purifications can decohere to standard quantum theory.
This result gives a suggested mechanism by which causality might arise from a non-causal but still quantum-informational theory, and highlights the subtle role of purity in arguments regarding the existence of hyper-decoherence.

This mechanism has a more direct physical interpretation than previous proposals for theories which hyper-decohere to quantum theory \cite{zyczkowski,dakic_cubes,gogioso_hypercubes,hefford_hypercubes}.
In $\qbox$, the fundamental nature of systems is postulated to be box-like with the observer having the power to implement any higher-order transformation, in other words, the observer has access to both past and future systems.
Hyper-decoherence imposes a reduction in power of the observer by only permitting them access to the future-evolving half of the box. 
There are in this sense, some similarities with decoherence from standard quantum to classical theory which imposes that the observer cannot isolate a system so that it is not continuously interacting with and so being decohered by the environment. 

Despite the possible interpretation in terms of the reduction in the power of the observer, one might naturally object that a hyper-decoherence process which permits an observer access to two systems, and then simply forbids the observer access to one of them, is too trivial and too far from the traditional notion of decoherence.
On the one hand, the possibility of hidden dimensions in this sense could be seen as unreasonable, and indeed, such a hyper-decoherence process was highlighted in \cite{lee_nogo} as precisely as the kind of process which the axioms of hyper-decoherence are intended to rule out. One could conclude that this says something about the fundamental difference between discarding in space and discarding in time, or one might instead conclude that discarding in time should also be ruled out.
The hyper-decoherence process of $\qbox$ motivates a careful re-examination of the axioms of hyper-decoherence.
One natural additional assumption that could be added is the preservation of the dimensionality of degrees of freedom before and after decoherence, and it appears unlikely that the hyper-decoherence from $\qbox$ will satisfy this.
Whether there might still exist post-quantum theories with dimensionality-preserving hyper-decoherence maps is left as a topic for future consideration.

On the other hand, in arguing for the legitimacy of this hyper-decoherence mechanism, it is interesting to realise that the seemingly minor move to allow the additional dimension to be an alternative temporal direction, as we do here, rather than an additional spatial dimension, surprisingly allows for the satisfaction of the axioms of hyper-decoherence, in particular, the purity co-preservation rule \footnote{Although, it should be noted that we took care to interpret purity co-preservation with regards to extremality in the deterministic state space (so extremality in the space of quantum channels). Purity co-preservation within the non-deterministic state space does not hold for our example for the same reasons outlined by Lee and Selby in \cite{lee_nogo} regarding the discarding of spatial dimensions.}.
Moreover, the additional temporal rather than spatial degrees of freedom lead naturally to higher-order quantum theory, and thus a toy model of features of quantum gravity  \cite{Hardy_2007}, precisely the sort of post-quantum theory from which we expect quantum theory to emerge in an appropriate limit.

Regarding future directions, if one accepts this hyper-decoherence process,  it is natural to ask whether there might be a generalised no-go theorem which establishes higher-order quantum theory at the top of any (even non-causal, and non-uniquely-purifiable) ladder of hyper-decoherence.
This could in turn provide an argument for higher-order quantum theory as the only natural candidate for a post-quantum theory from which traditional quantum theory could emerge.
Regarding the existence of such a ladder, most of the steps of the proof presented here are sufficiently general to apply to any underlying process theory $\C$.
In particular, the construction of $\D(\C)$ is theory independent so one can happily consider iterating it to get $\D(\D(\C))$ etc.
Although its multi-environment structure relies upon the definition of process matrices, one could either look specifically at the case $\C=\CP$ and attempt to generalise the permissible discarding effects, or take the general definition of supermaps on higher-order theories outlined in \cite{wilson_locality,hefford_supermaps} to provide the necessary effects in a theory independent fashion.
At that point, one must contend with generalising the hyper-decoherence map, which depends on the unique discarding effect of standard quantum theory to form the completely depolarising map.
It may be that some preferred element of the multi-environment structure should be used or alternatively there might exist a family of possible hyper-decoherence maps, one for each element. Either way, the satisfaction of the axioms of hyper-decoherence may not come for free, in particular, the most subtle axiom \hyperref[ax:3]{Ax3} regarding the copreservation of purity which in the case of $\qbox$ relies on special properties of standard quantum theory.
The end result would be a proof that \textit{any} process theory $\C$ can hyper-decohere from $\D(\C)$ and thus we could build the entire tower of higher-order quantum theory with each order decohering into the one below it.

More speculative, is the possibility of proving that \text{any} theory hyper-decohering to higher-order quantum theory would have to be contained within higher-order quantum theory.
The existing no-go theorem relies heavily on uniqueness of purifications, and so, new techniques and re-axiomatisations might be necessary.
One natural approach could be to adapt the interpretation of purity away from convex extremality, for instance to consider a state to be pure if any extension of it factorises \cite{chiribella2015distinguishabilitycopiabilityprogramsgeneral}.
Another approach could be to consider Yoneda-like techniques for finding internal representative maps of functions with properties from the outside. 
While functions satisfying naturality properties on a process theory may not have internal representatives (being higher-order maps), families of functions satisfying naturality properties on already higher-order process theories, particularly those which are compact closed, are often forced to have internal representatives \cite{wilson2026supermapsgeneralisedtheories} suggesting a sense in which it is harder to step outside from a genuinely higher-order theory.

Another possible route for future work would be to look at higher-order interference, a commonly postulated feature of theories which are more coherent than quantum theory \cite{sorkin_measure,sorkin_measure2,zyczkowski,lee_interference}.
The existence of a hyper-decoherence from $\qbox$ suggests the possibility that $\qbox$ and so also higher-order quantum theory might support higher-order interference phenomena.

In summary, the existence of a hyper-decoherence map from higher-order quantum theory invites a rich line of enquiry into how quantum theory might appear from these toy models of quantum gravity and ultimately from quantum gravity itself while elucidating the care with which issues around purity and causality must be treated when searching for post-quantum theories.

\begin{acknowledgments}
    This work was funded by the French National Research Agency (ANR) within the framework of ``Plan France 2030'', under the research projects EPIQ ANR-22-PETQ-0007 and HQI-R\&D ANR-22-PNCQ-0002.
\end{acknowledgments}

\bibliography{bibliography}

\appendix

\section{No Superluminal Signalling in $\qbox$}\label{app:superluminal}
In this section we show that the theory $\qbox$ is no-superluminal signalling, meaning that for any bipartite state $\psi$ and pair of effects $\Ground_1, \Ground_2$ we have the following
\begin{equation*}
    \input{figs/no_super_5.tikz} \ =  \   \input{figs/no_super_4.tikz}.
\end{equation*}
The satisfaction of this theorem in the case in which the left and right systems are atomic, meaning that $\psi$ simply has the form $[L,L'] \otimes [R,R']$, is proven already in \cite{Wilson_2021}.
To extend this to general bipartite states in $\qbox$ we must check for arbitrary states of type $(\otimes_i [L_i,L_i']) \otimes (\otimes_k [R_k,R_k'])$. For this case, note that any (arbitrary-arity) multi-partite non-signalling channel in any GPT can be rewritten as an affine linear combination of localised quantum channels \cite{Cavalcanti_2022}.
Therefore, since any bipartite state in $\qbox$ is an arbitrary-arity multi-partite non-signalling channel, it can be rewritten as an affine linear combination of its left and right parts.

As a result, we have that for any bipartite state $\psi$ and pair of effects $\Ground_1, \Ground_2$, 
\begin{equation*}
    \input{figs/no_super_5.tikz} \ =  \ \sum_{ik}\alpha_{ik}  \input{figs/super_app_1.tikz}  \ =  \ \sum_{ik}\alpha_{ik}   \input{figs/super_app_2.tikz}
\end{equation*}
\begin{equation*}
  \ =  \  \sum_{ik}\alpha_{ik}    \input{figs/super_app_3.tikz} \ =  \   \input{figs/no_super_4.tikz},
\end{equation*}
and so there can be no superluminal signalling via states in $\qbox$.

\section{Purifications in $\qbox$}\label{app:purifications}

In this section we will show that purifications in $\qbox$ are not unique.
Consider the following two non-signalling channels i.e.\ bipartite states when interpreted in $\qbox$,
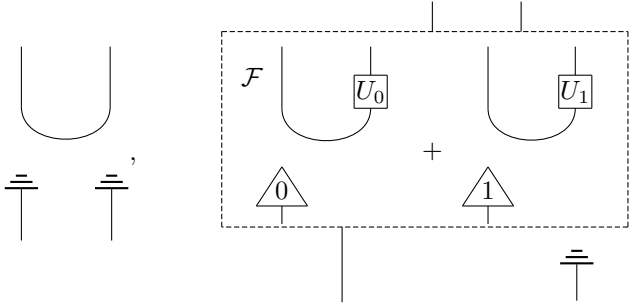
\begin{equation*}
    \input{figs/unique_1.tikz},   \quad  \quad  \quad    \input{figs/unique_2double.tikz}
\end{equation*}
where $U_0$ and $U_1$ are taken to be two distinct unitaries $U_0 \neq U_1$ up to a global phase and $\mathcal{F}$ is taken to be the doubling functor from the process theory of linear maps between complex vectors spaces to the process theory of CP maps. This functor acts as the identity on objects and acts on morphisms by $\mathcal{F}(U)(\rho) =U \rho U^{\dagger}$.
Note that both of these bipartite states in $\qbox$ possess the same reduced state
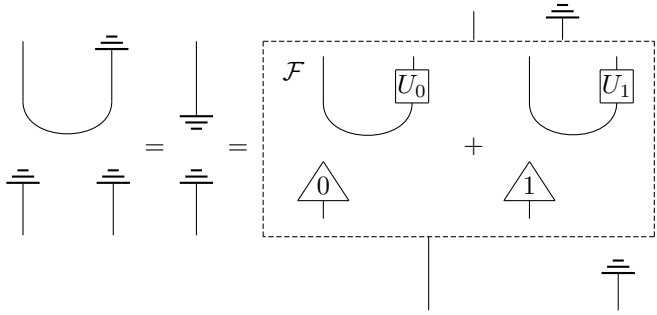
\begin{equation*}
    \input{figs/unique_4.tikz} \ = \      \begin{tikzpicture}
	\begin{pgfonlayer}{nodelayer}
		\node [style=none] (19) at (0, -1.475) {};
		\node [style=smalldiscarding] (20) at (0, -0.5) {};
		\node [style=smallbackdiscard] (21) at (0, 0.4) {};
		\node [style=none] (22) at (0, 1.75) {};
	\end{pgfonlayer}
	\begin{pgfonlayer}{edgelayer}
		\draw (19.center) to (20);
		\draw (22.center) to (21);
	\end{pgfonlayer}
\end{tikzpicture}
 \ = \    \input{figs/unique_6double.tikz}
\end{equation*}
If purifications were unique up to a reversible process on the environment, then there would exist some reversible comb such that
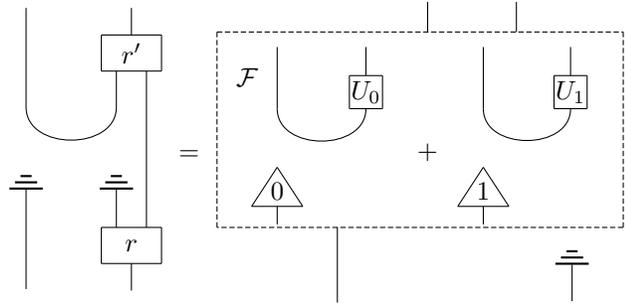
\begin{equation*}
    \input{figs/unique_8.tikz} \ = \       \input{figs/unique_2double.tikz} 
\end{equation*}
However, this would entail that
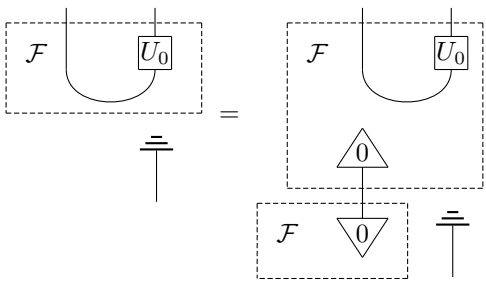
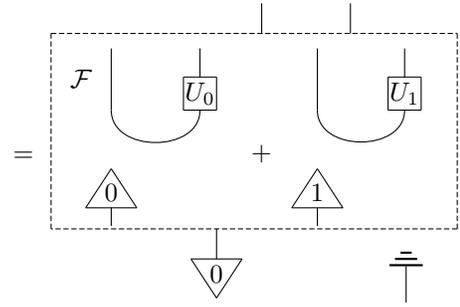
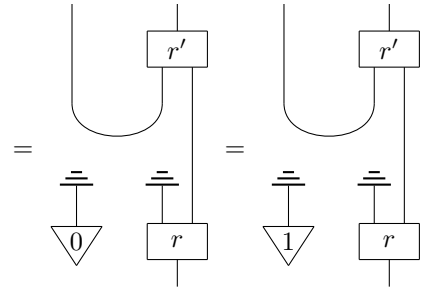
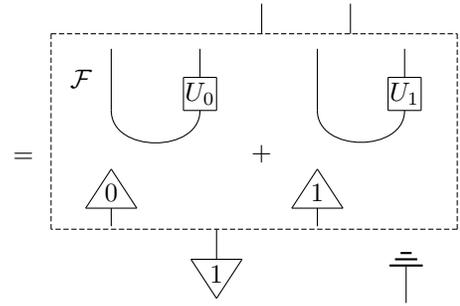
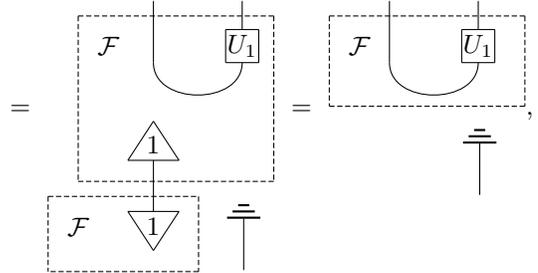
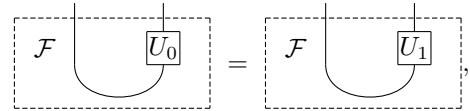
\begin{equation*}
    \input{figs/unique_9adouble.tikz} \ = \        \input{figs/unique_9bdouble.tikz}  
\end{equation*}
\begin{equation*}
    \ = \   \input{figs/unique_9c.tikz}  
\end{equation*}
\begin{equation*}
 \ = \        \input{figs/unique_10.tikz}     \ = \   \input{figs/unique_11.tikz} 
\end{equation*}

\begin{equation*}
\ = \        \input{figs/unique_12adouble_new.tikz}
\end{equation*}
    
\begin{equation*}
\ = \        \input{figs/unique_12adouble.tikz} \ = \   \input{figs/unique_12bdouble.tikz},    
\end{equation*}
and so
\begin{equation*}
    \input{figs/unique_13pre.tikz} \ = \     \input{figs/unique_14pre.tikz}    ,
\end{equation*}
and up to a global phase
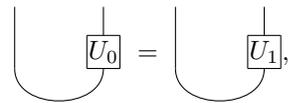
\begin{equation*}
    \input{figs/unique_13.tikz} \ = \     \input{figs/unique_14.tikz}    ,
\end{equation*}
which is a contradiction with the initial assumption that $U_0 \neq U_1$ up to a global phase. Therefore, purifications in $\qbox$ cannot be unique.

\section{Equivalence of Hyper-decohered Systems of $\qbox$ with Quantum Theory}\label{app:equivalence}
\begin{proof}
    Denote by $\mathsf{Hypdec}$ the category of hyper-decohered systems, that is the full sub-process theory of $\Split(\qbox)$ spanned by systems of the form $(\otimes_i [H_i,H_i],\hypdec)$.
    As outlined in the main text, we will show that there is a fully faithful and essentially surjective on objects monoidal functor (an equivalence of process theories) $F:\mathsf{Hypdec}\morph{}\CPTP$.
    Define $F(\otimes_i [H_i,H_i],\hypdec) := \otimes_i H_i$ on systems and on processes as,
    \begin{equation*}
        \input{figs/equiv_proof1_ico.tikz}
    \end{equation*}
    Again, since this is the partial application of a superchannel on two CPTP maps, the result is automatically CPTP.    
    It is straightforward to see that $F$ is a functor: the identity on $(\otimes_i [H_i,H_i],\hypdec)$ is given by $\hypdec$ which is sent to $1_{\otimes_i H_i}$ by $F$, and composition is preserved because the morphisms of $\mathsf{Hypdec}$ are hyper-decohered and thus take the following form.
    \begin{equation*}
        \input{figs/equiv_proof2_ico.tikz}
    \end{equation*}
    $F$ preserves tensor products because,
    \begin{align*}
        & F\big((\otimes_i[H_i,H_i],\hypdec)\otimes(\otimes_j[K_j,K_j],\hypdec)\big) \\ 
        & = F\big((\otimes_i[H_i,H_i])\otimes(\otimes_j[K_j,K_j]),\hypdec\big) \\
        & = (\otimes_i H_i) \otimes (\otimes_j K_j) \\
        & = F(\otimes_i[H_i,H_i],\hypdec)\otimes F(\otimes_j[K_j,K_j],\hypdec)
    \end{align*}
    $F$ is clearly essentially surjective on objects: any Hilbert space $H$ is the image of $[H,H]$ under $F$.
    $F$ is full because each $f:\otimes_i H_i\morph{}\otimes_j K_j$ is the image of the supermap,
    \begin{equation*}
        \input{figs/equiv_proof3_ico.tikz}: \begin{array}{l} (\otimes_i [H_i,H_i],\hypdec) \\ \morph{}(\otimes_j [K_j,K_j],\hypdec) \end{array}
    \end{equation*}
    Finally, $F$ is faithful, since if two supermaps $S,T: (\otimes_i [H_i,H_i],\hypdec)\morph{}(\otimes_j [K_j,K_j],\hypdec)$ are such that their hyper-decohered versions are equal,
    \begin{equation*}
        \input{figs/equiv_proof4_ico.tikz}
    \end{equation*}
    then it follows that the original supermaps were equal.
    \begin{align*}
       &  \input{figs/equiv_proof5_ico_a.tikz} \\
        = \ &   \input{figs/equiv_proof5_ico_b.tikz}  
    \end{align*}
    This completes the proof, though it is also straightforward to construct the inverse functor $G:\CPTP \morph{} \mathsf{Hypdec}$.
    On objects $G(H) = ([H,H],\hypdec)$ and on morphisms $f:H\morph{}K$ is sent to the supermap 
    \[  \input{figs/equiv_proof6_ico.tikz}  :([H,H],\hypdec)\morph{}([K,K],\hypdec)\]
    It is fairly easy to write down the required natural isomorphisms $\varepsilon:FG\Rightarrow 1_{\CPTP}$ and $\eta:1_{\mathsf{Hypdec}}\Rightarrow GF$.
    The components of $\varepsilon$ are the identity, while those of $\eta$ are given by the supermaps $\otimes_i[H_i, H_i] \morph{} [\otimes_i H_i , \otimes_i H_i]$ that are the identity on top and prepare the maximally mixed state on the bottom. 
\end{proof}

\end{document}